\documentclass{article}              

\usepackage{bm,bbm}
\usepackage[square,numbers]{natbib}
\usepackage{graphicx}
\usepackage{amsmath,amsfonts,amssymb}
\usepackage{authblk}

\usepackage[colorlinks=true,breaklinks=true,bookmarks=true,urlcolor=blue,
     citecolor=blue,linkcolor=blue,bookmarksopen=false,draft=false]{hyperref}

\def\EMAIL#1{\href{mailto:#1}{#1}}
\def\URL#1{\href{#1}{#1}}

\def\qed{\hfill \vrule height 7pt width 7pt depth 0pt\medskip}

\newcommand{\ba}{\begin{array}}
\newcommand{\ea}{\end{array}}

\newcommand{\be}{\begin{equation}}
\newcommand{\ee}{\end{equation}}

\newcommand{\mc}{\mathcal}

\newcommand{\ov}{\overline}

\newcommand{\R}{\mathbb{R}}

\renewcommand{\P}{\mathbb{P}}

\newcommand{\se}{\text{ if }}

\def\R{\mathbb{R}}

\def\P{\mathbb{P}}

\newtheorem{definition}{Definition}
\newtheorem{example}{Example}
\newtheorem{corollary}{Corollary}
\newtheorem{lemma}{Lemma}
\newtheorem{proposition}{Proposition}
\newtheorem{remark}{Remark}
\newtheorem{theorem}{Theorem}

\title{Separable Games}
\author[1]{Laura Arditti}
\author[1,2]{Giacomo Como}
\author[1]{Fabio Fagnani}
\affil[1]{\footnotesize Dipartimento di Scienze Matematiche, Politecnico di Torino, Corso Duca degli Abruzzi 24, 10129, Torino, Italy}
\affil[2]{\footnotesize Department of Automatic Control, Lund University, BOX 118, SE-22100, Lund, Sweden}
\date{\small \EMAIL{laura.arditti@polito.it}, \EMAIL{giacomo.como@polito.it}, \EMAIL{fabio.fagnani@polito.it}\\ 
	\URL{https://laura-arditti.github.io/}, \URL{https://staff.polito.it/giacomo.como/},  \URL{http://calvino.polito.it/~fagnani/}}

\begin{document}

\maketitle

\begin{abstract}
		We present the notion of separable game with respect to a forward directed hypergraph (FDH-graph), which refines and generalizes that of graphical game. First, we show that there exists a minimal FDH-graph with respect to which a game is separable, providing a minimal complexity description for the game. Then, we prove a symmetry property of the minimal FDH-graph of potential games and we describe how it reflects to a decomposition of the potential function in terms of local functions. In particular, these last results strengthen the ones recently proved for graphical potential games. Finally, we study the interplay between separability and the decomposition of finite games in their harmonic and potential components, characterizing the separability properties of both such components.
\end{abstract}

\section{Introduction}\label{sec:introduction}
Network games have recently emerged as a unified framework for modeling interactions in many social and economic settings \cite{Morris:2000,Jackson:2008,Easley.Kleinberg:2010,Jackson.Zenou:2015,Bramoulle.Kranton:2016}. They allow one to analyze the emergence of phenomena such as peer effects, technology adoption, spread of ideas and innovation, consensus formation, diffusion of crime or education, and commitment to public good \cite{Ballester.ea:2006,Bramoulle:2007,Bramoulle.Kranton:2007,Young:2009,Montanari.Saberi:2010,Acemoglu.ea:2015, Allouch:2015}.  What characterizes such network interactions is often a notion of locality, as the decisions of individuals are affected by the actions of their friends, colleagues, peers, or competitors. This may be formailized by the notion of graphical game \cite{Kearns.ea:2001} (see also \cite[Chapter 7]{Nisan.ea:2007}) where players are identified with nodes of a graph and the utility of each player depends only on her own action and the action of players corresponding to her neighbors in such graph. Graphical games are also a natural model in engineering and computer science to describe the interaction patterns of multi-agent systems and provide a powerful tool to design distributed algorithms \cite{Daskalakis.Papadimitriou:2006,Papadimitriou.Roughgarden:2008}.


In a graphical game, the way the utility of a player depends on the actions played by her neighbor players is a key feature that plays a crucial role in the analysis of the game (e.g., Nash equilibria, existence of a potential). In classical coordination or anti-coordination games such dependence can be seen as the sum of pairwise interactions with each single neighbor player. In other cases, as in the best-shot public good games, instead such decomposition is not possible. 

While this finer structure of the utility functions cannot be addressed within the notion of graphical game, it is at the core of the theory of separable games, that is the main focus of this paper. Separable game is a more refined notion than graphical game and allows for a finer description of the dependence pattern among the players in a game. Of particular relevance is the case of potential games for which our theory gives an exact correspondence between the separability property of a game and the decomposition of the potential function as a sum of local functions. As a corollary, we recover a remarkable result in \cite{Babichenko.Tamuz:2016} where authors prove that every potential graphical game admits a potential reflecting the graphical structure, precisely, one that can be decomposed as a sum of terms defined on the maximal cliques of the graph. While the result in \cite{Babichenko.Tamuz:2016} relies on the Hammersley-Clifford theorem for Markov random fields  \cite[Chapter 3]{Lauritzen:1996}, our derivation is completely autonomous. We actually prove that our results yield an alternative proof of the Hammersley-Clifford theorem.

The contribution of this paper is as follows. First of all, we introduce a novel notion of separability of a game with respect to a forward directed hypergraph (FDH-graph). The proposed notion of separable game encompasses and refines that of graphical game.  A key feature of our approach is that we consider games up to strategic equivalence, meaning that we are only concerned with variations of the utility of a player when she modifies her action rather then their absolute values. This approach is natural in technological contexts where the game rather then an intrinsic model is the result of an explicit design. In fact most classical game theoretic concepts such as domination, Nash equilibrium, correlated equilibrium, best-response dynamics or the logit dynamics are invariant with respect to strategic equivalence. 

Our first main result is Theorem \ref{theo:minimal-class} stating the existence of a minimal FDH-graph with respect to which a game is separable. We prove this result as a consequence of an intermediate technical result, Lemma \ref{lemma:separable-intersection}, concerning separability of a single function with respect to an undirected hypergraph (H-graph). 

Our second main contribution consists in investigating the separability of potential games. In particular, in Theorem \ref{theorem:potential} we show that the minimal FDH-graph with respect to which a potential game is separable has a symmetry property and is completely determined by the H-graph describing the minimal separability of the potential function itself. This result implies Theorems 4.1 and 4.4 in \cite{Babichenko.Tamuz:2016} as well as the Hammersley-Clifford theorem for Markov random fields. 

Finally, we analyze how the proposed notion of separability interacts with the decomposition of a game in terms of potential and harmonic parts introduced in \cite{Candogan.ea:2011}. In Theorem \ref{theo:main} we characterize the minimal FDH-graphs of the potential and harmonic parts, respectively, showing that they can sometimes be larger than the minimal FDH-graph of the original game. At the pure graphical level, this corresponds to a sort of ``interaction enlargement'', namely the appearance of short range strategic interactions that are hidden in the original game but emerge in the decomposition.

We conclude this introduction with a brief outline of the rest of this paper. Section \ref{sec:preliminaries} is devoted to introducing all the necessary graph theoretic and game theoretic concepts, including the definition of graphical game. In Section \ref{sec:separable-games} we introduce the concepts of separability of a function with respect to a H-graph and of a game with respect to a FDH-graph. We then prove our main technical result Lemma \ref{lemma:separable-intersection} that implies the existence of a minimal separating H-graph for a function (Proposition \ref{proposition:minimal-separation}). From this, we derive our first fundamental result, Theorem \ref{theo:minimal-class}, on the existence of a minimal separating FDH-graph for a game. Section \ref{sec:potential} is devoted to potential games and it contains Theorem \ref{theorem:potential} that describes the relation between the minimal separating FDH-graph associated with the game and the minimal H-graph describing the separation of the potential function.  We then show how this result
implies results on graphical potential games and on Markov random fields. In Section \ref{sec:main-results2} we present Theorem \ref{theo:main} that analyzes the relation between the separability of a game and that of its potential and harmonic components. The paper ends with a conclusive Section \ref{sec:conclusion}.

\section{Technical background}\label{sec:preliminaries}
In this section, we introduce all the necessary graph theoretic and game theoretic concepts, to be used in the rest of the paper.

\subsection{Graph-theoretic preliminaries} \label{sec:graphs} We start by introducing some basic graph-theoretic definitions and notation. We shall begin with directed graphs, then consider undirected and forward directed hypergraphs, and finally show how these notions are related.

\subsubsection{Directed graphs}A \emph{directed graph} $\mc G=(\mc V,\mc E)$ is the pair of a finite node set $\mc V$ and of a link set $\mc E\subseteq\mc V\times\mc V$, where a link $(i,j)$ in $\mc E$ is meant as directed from its tail node $i$ to its head node $j$. Throughout the paper, we shall consider directed graphs containing no self-loops, i.e., such that $(i,i)$ does not belong to the link set $\mc E$ for any node $i$ in $\mc V$, and refer to them simply as graphs. We shall denote by $\mc N_{i}=\{j\in\mc V:\,(i,j)\in\mc E\}$ and $\mc N^{\bullet}_i=\mc N_{i}\cup\{i\}$ the open and, respectively, closed out-neighborhoods of a node $i$ in a graph $\mc G=(\mc V,\mc E)$. The intersection of two graphs $\mc G_1=(\mc V,\mc E_1)$ and $\mc G_2=(\mc V,\mc E_2)$ with the same node set $\mc V$ is the graph $\mc G_1\sqcap\mc G_2=(\mc V,\mc E)$ with link set $\mc E=\mc E_1\cap\mc E_2$. We shall say that $\mc G_1=(\mc V,\mc E_1)$ is a subgraph of $\mc G_2=(\mc V,\mc E_2)$ and write $\mc G_1\preceq\mc G_2$,  if $\mc E_1\subseteq\mc E_2$, equivalently, if $\mc G_1\sqcap\mc G_2=\mc G_1$. We shall consider undirected graphs as a special case of graphs $\mc G=(\mc V,\mc E)$ such that there exists a directed link $(i,j)$ in $\mc E$ if and only if the reversed directed link $(j,i)$ in $\mc E$ exists as well. 

\subsubsection{Hypergraphs} A \emph{hypergraph} (shortly, a \emph{H-graph}) is the pair $\mc H=(\mc V,\mc L)$ of a finite node set $\mc V$ and of a set $\mc L$ of  undirected hyperlinks, each of which is a nonempty subset of nodes  \cite{Bretto:2013}.  A H-graph $\mc H=(\mc V,\mc L)$ is referred to as \emph{simple} if no undirected hyperlink $\mc J$ in $\mc L$ is strictly contained in another undirected hyperlink $\mc K$ in $\mc L$. The simple H-graph associated to a H-graph $\mc H=(\mc V,\mc L)$ is the H-graph $\ov{\mc H}=(\mc V,\ov{\mc L})$ with set of undirected hyperlinks
$$\ov{\mc L}=\{\mc J\in\mc L:\nexists \mc K\in\mc L\text{ s.t.~}\mc K\supsetneq\mc J\}\,.$$  
We shall say that a H-graph $\mc H_1 = (\mc V, \mc L_1)$ is a sub-H-graph of another H-graph $\mc H_2 = (\mc V, \mc L_2)$ and write $\mc H_1\preceq\mc H_2$ if, for every undirected hyperlink $\mc J_1$ in $\mc L_1$, there exists some undirected hyperlink $\mc J_2$ in $\mc L_2$ such that $\mc J_1\subseteq\mc J_2$. 
Notice that both $\mc H\preceq\ov{\mc H}$ and $\ov{\mc H}\preceq\mc H$ and that in fact $\ov{\mc H}$ is the only simple H-graph with this property. 
Also, observe that $\mc H_1\preceq\mc H_2$ if and only if $\ov{\mc H_1}\preceq\ov{\mc H_2}$. Given two H-graphs  $\mc H_1 = (\mc V, \mc L_1)$ and $\mc H_2 = (\mc V, \mc L_2)$ with the same node set $\mc V$, the intersection of $\mc H_1$ and $\mc H_2$ is  the H-graph $\mc H_1\sqcap\mc H_2=(\mc V,\mc L)$ with set of undirected hyperlinks $$\mc L=\{\mc J=\mc J_1\cap\mc J_2\,:\ \mc J_1\in\mc L_1,\,\mc J_2\in\mc L_2\}\,.$$
The union of $\mc H_1$ and $\mc H_2$ is instead defined as the H-graph $\mc H_1\sqcup\mc H_2=(\mc V,\mc L_1\cup\mc L_2)$.

\subsubsection{Forward directed hypergraphs} The following is a generalization of both notions of graphs and H-graphs above.  A \emph{forward directed hypergraph} (\emph{FDH-graph}) \cite{Gallo.ea:1993} is the pair $\mc F = (\mc V, \mc D)$ of a finite node set $\mc V$ and of a finite hyperlink set $\mc D$, where each hyperlink $d=(i,\mc J)$ in $\mc D$ is the ordered pair of a node $i$ in $\mc V$ (to be referred to as its tail node) and a nonempty subset of head nodes $\mc J\subseteq\mc V\setminus\{i\}$ (to be referred to as the hyperlink's head set). A FDH-graph $\mc F = (\mc V, \mc D)$ is simple if for every hyperlink $(i,\mc J)$ in $\mc D$, there exists no other hyperlink $(i,\mc K)$ in $\mc D$ such that $\mc J\subsetneq\mc K$. The simple FDH-graph associated to a FDH-graph $\mc F = (\mc V, \mc D)$  is the FDH-graph $\ov{\mc F} = (\mc V,\ov{\mc D})$  with hyperlink set $$\ov{\mc D}=\left\{(i,\mc J)\in\mc D:\,\nexists (i,\mc K)\in\mc D\text{ s.t. }\mc K\supsetneq\mc J\right\}\subseteq\mc D$$ containing only hyperlinks in $\mc D$ with maximal head node set. 
We shall say that a FDH-graph $\mc F_1 = (\mc V, \mc D_1)$ is a sub-FDH-graph of another FDH-graph $\mc F_2 = (\mc V, \mc D_2)$ and write $\mc F_1\preceq\mc F_2$ if for every hyperlink $(i,\mc J_1)$ in $\mc D_1$ there exists some hyperlink $(i,\mc J_2)$ in $\mc D_2$ such that $\mc J_1\subseteq\mc J_2$. Notice that both $\mc F\preceq\ov{\mc F}$ and $\ov{\mc F}\preceq\mc F$ and that in fact $\ov{\mc F}$ is the only FDH-graph with this property. Also, observe that $\mc F_1\preceq\mc F_2$ if and only if $\ov{\mc F_1}\preceq\ov{\mc F_2}$. The intersection of two FDH-graphs  $\mc F_1 = (\mc V, \mc D_1)$ and $\mc F_2 = (\mc V, \mc D_2)$ with the same node set $\mc V$ is the FDH-graph $\mc F_1\sqcap\mc F_2=(\mc V,\mc D)$ with set of directed hyperlinks \be\label{intersection-hyperlinks}\mc D=\{(i,\mc J):\, \exists (i,\mc J_1)\in\mc D_1, (i,\mc J_2)\in\mc D_2\text{ s.t. }\mc J_1\cap\mc J_2=\mc J\}\,.\ee

\subsubsection{How the various concepts are related}
A directed graph $\mc G=(\mc V,\mc E)$ can naturally be seen as a FDH-graph on $\mc V$ whose hyperlinks are the original links in the graph $(i,j)$ interpreted as $(i,\{j\})$; with slight abuse of notation in the following we will identify such FDH-graph with the original graph $\mc G$. 

There is another natural way of relating directed graphs and FDH-graphs, that will play a key role in the following. On the one hand, to a directed graph $\mc G=(\mc V,\mc E)$, we can associate the FDH-graph
\be\label{FG}\mc F^{\mc G}=(\mc V,\mc D^{\mc G}),\qquad \mc D^{\mc G}=\{(i, \mc N_i)\,|\, i\in\mc V\}\,,\ee
on the other hand, for a FDH-graph $\mc F=(\mc V,\mc D)$ we can consider the directed graph 
\be\label{GF}\mc G^{\mc F}=(\mc V,\mc E^{\mc F}),\qquad \mc E^{\mc F}=\{(i, j)\,|\, \exists (i,\mc J)\in\mc D,\, j\in\mc J\}\,.
\ee
Observe that the following relations hold true:
\be\label{GFrelation} \mc G=\mc G^{\left(\mc F^{\mc G}\right)},\quad \mc F\preceq \mc F^{\left(\mc G^{\mc F}\right)}\,.\ee

Analogously, on the one hand to every H-graph $\mc H=(\mc V,\mc L)$ we can associate the  FDH-graph 
\be\label{FH}\mc F^{\mc H}=(\mc V,\mc D^{\mc H}),\qquad \mc D^{\mc H}=\{(i, \mc J)\,|\, i\in\mc V\setminus\mc J, \, \{i\}\cup\mc  J\in\mc L\}\,,\ee
on the other hand to every FHD-graph $\mc F=(\mc V,\mc D)$ we can associate the H-graph 
\be\label{HF}\mc H^{\mc F}=(\mc V,\mc L^{\mc F}),\quad \mc L^{\mc F}=\{\{i\}\cup\mc J:\,(i,\mc J)\in\mc D\}\,,\ee
so  that the following relations hold true:
\be\label{HFrelation} \mc H=\mc H^{\left(\mc F^{\mc H}\right)},\quad \mc F\preceq \mc F^{\left(\mc H^{\mc F}\right)}\,.\ee
FDH-graphs $\mc F=(\mc V,\mc D)$ that are derived from H-graphs in the sense that $\mc F=\mc F^{\mc H}$ for some H-graph $\mc H$ are referred to as undirected as they are characterized by the property that 
$$(i,\mc J)\in\mc D\;\Rightarrow\; (j,\{i\}\cup\mc J\setminus\{j\})\in\mc D\,,\quad\forall j\in \mc J\,.$$ 
The FDH-graph $\mc F^{\left(\mc H^{\mc F}\right)}$ is called the underlying undirected FDH graph associated with $\mc F$ and, for notational simplicity, it is denoted by $\mc F^{\leftrightarrow}=(\mc V,\mc D^{\leftrightarrow})$. The set of directed hyperlinks can be characterized as
\be\label{underlying-def}\mc D^{\leftrightarrow}=\{(i,\mc J):\,\{i\}\cup\mc J=\{h\}\cup\mc K\text{ for some }(h,\mc K)\in\mc D\}\,.\ee

\subsection{Game-theoretic preliminaries}\label{sec:games}
Throughout the paper, we shall consider strategic form games with 
finite nonempty player set $\mc V$ and a nonempty action set $\mc A_i$ for each player $i$ in $\mathcal{V}$.  We shall denote by $\mathcal{X}= \prod_{i \in \mathcal{V}} \mc A_i$ the space of all players' strategy profiles and, for every player $i$ in $\mc V$, we let $\mathcal{X}_{-i}= \prod_{j \in \mathcal{V}\setminus\{i\}} \mc A_j$ to be the set of strategy profiles of all players except for player $i$. 
As customary, for a strategy profile $x$ in $\mc X$, the strategy profile of all players except for $i$ is denoted by $x_{-i}$ in $\mc X_{-i}$, while $x_{\mc J}$ denotes the strategy profiles restricted to a subset $\mc J\subseteq \mc V$. We shall refer to two strategy profiles $x$ and $y$  in $\mc X$ as  $i$-comparable and write $x\sim_i y$ when $x_{-i}=y_{-i}$, i.e., when $x$ and $y$ coincide except for possibly in their $i$-th entry.  

We let each player $i$ in $\mathcal{V}$ be equipped with a utility function $u_i:\mathcal{X}\to \mathbb{R}\,.$ We shall identify a game with player set $\mc V$ and strategy profile space $\mc X$ with the vector $u$ assembling all the players' utilities. Notice that, in this way, the set of all games with player set $\mc V$ and strategy profile space $\mc X$, to be denoted as $\mc U$, is isomorphic to the vector space $\R^{\mc V\times\mc X}$.

A game $u$ is referred to as \emph{non-strategic} if the utility of each player $i$ in $\mc V$ does not depend on her own action,  i.e., if
\be\label{eq:non-strategic}	u_i(x)=u_i(y)\,,\qquad\forall x,y\in\mc X\text{ s.t.~}y\sim_i x\,.\ee
Two games $u$ and $\tilde u$ are referred to as \emph{strategically equivalent} if their difference is a non-strategic game. Strategic equivalence is an equivalence relation on $\mc U$. In this paper, we will focus on properties of a game that are invariant with respect to strategic equivalence. 

A game $u$ in $\mc U$ is referred to as an (exact) \emph{potential game} \cite{Monderer.Shapley:1996} if there exists a \emph{potential function} $\phi:\mathcal{X}\to\mathbb{R}$ such that 
\be\label{eq:potential}u_i(x)-u_i(y)=\phi(x)-\phi(y)\,,\ee 
for every player $i$ in $\mc V$ and every pair of $i$-comparable strategy profiles $x\sim_i y$ in $\mc X$. Observe that \eqref{eq:potential} may be rewritten as $u_i(x)-\phi(x)=u_i(y)-\phi(y)$ for every $x\sim_i y$ in $\mc X$ which is in turn equivalent to that $u_i(x)-\phi(x)=n_i(x_{-i})$ does not depend on $x_i$ for every strategy profile $x$ in $\mc X$. Hence, $u$ in $\mc U$ is a potential game with potential $\phi(x)$ if and only if it is strategically equivalent to a game $u^{\phi}$ in $\mc U$ with utilities $u^{\phi}_i(x)=\phi(x)$ for every player $i$ in $\mc V$.

\subsection{Graphical games} \label{sec:graphical-games}

Graphical games \cite{Kearns.ea:2001} are defined, with respect to a fixed graph $\mc G$, imposing that the utility of each player $i$ only depends on the action configuration in its closed neighborhood. We slightly depart from this and we assume that this holds up to non-strategic parts. This allows us a much more compact and clear presentation of our results. A similar point of view has been already considered in the literature \cite{Babichenko.Tamuz:2016}. The formal definition is the following.
%
%
%
%
\begin{definition}\label{def:graphical-game}
Given a graph  $\mc G=(\mc V,\mc E)$, a game $u$ is said to be \emph{graphical} with respect to $\mc G$, or to be a $\mc G$-game, if 
	the utility of each player $i$ in $\mc V$ 
	can be decomposed as  
	\be\label{eq:utility-decomposition-graphical}u_i(x)=v_i(x_i,x_{\mc N_i})+n_i(x_{-i})\,,\ee
	where  $v_i:\mc A_i\times\prod_{j\in\mc N_i}\mc A_j\to\R$ is a function that depends on the action of player $i$ and of players in the subset $\mc N_i$ only, while $n_i:\mc X_{-i}\to\R$ is a non-strategic component that does not depend on the action of player $i$. 
\end{definition} \medskip
By definition, the notion of graphicality introduced above is invariant with respect to strategic equivalence.

Notice that if a game $u$ is graphical with respect to two graphs $\mc G_1 = (\mathcal{V}, \mathcal{E}_1)$ and $\mathcal{G}_2 = (\mathcal{V}, \mc{E}_2)$, it is also graphical with respect to their intersection $\mc G_1\sqcap\mc G_2$. Since every game is trivially graphical on the complete graph on $\mc V$, we can conclude that to each game $u$ in $\mc U$ one can always associate the smallest graph on which $u$ is graphical. We shall refer to such graph as the \emph{graph of the game} $u$ and denote it as $\mc G_u$.

In fact, important classes of graphical games have finer separability properties.  
In particular, for a given graph $\mc G=(\mc V,\mc E)$, a \emph{pairwise-separable network game} (cf.~\cite{Daskalakis.Papadimitriou:2009,Cai.Daskalakis:2011}) on $\mc G$ is such that the utility of player $i$ in $\mc V$ can be decomposed in the form 
\be\label{pairwise-game}
u_i(x) = \sum_{j \in \mc N_{i}} u_{ij}(x_i,x_j) \qquad \forall x \in \mathcal{X}\,,
\ee
where $u_{ij}:\mc A_i\times\mc A_j\to\R$ for $i,j$ in $\mc E$. 
These games are also known as polymatrix games on $\mc G$ \cite{Yanovskaya:1968} and are a special case of $\mc G$-games.
In the special case when $\mc G$ is undirected, 
such games can be interpreted as if each pair of players $i,j$ connected by a link were involved in a two-player game having utility functions, respectively, $u_{ij}(x_i,x_j)$ and $u_{ji}(x_j,x_i)$. 
Each player $i$ in $\mc V$  chooses the same action $x_i$ in $\mc A_i$ for all the two-player games it is engaged in and gets a utility that is the aggregate of the utilities of all such games.

In the next section, we will study a general notion of separability of games for which the pairwise-separable network games are a special case. We end this section with two examples.

\begin{example}[Network coordination game]\label{example:coord.game} For a graph $\mc G=(\mc V,\mc E)$, a \emph{network coordination game} on $\mc G$ is a game $u$ where every player $i$ in $\mc V$ has binary action set $\mc A_i=\{0,1\}$ and utility function
\be\label{eq:coord.game}u_i(x)=\sum_{j\in\mc N_i}\zeta(x_i,x_j)\,,\ee
where $\zeta(x_i,x_j)=\zeta(x_j,x_i)$ is a symmetric function such that $\zeta(0,0)\ge\zeta(0,1)=\zeta(1,0)$ and $\zeta(1,1)\ge\zeta(0,1)=\zeta(1,0)$.
Clearly, every network coordination game with utilities \eqref{eq:coord.game} is a pairwise-separable game on $\mc G$. Moreover, if the graph $\mc G$ is undirected, a network coordination game on $\mc G$ is a potential game with potential function 
$$\phi(x)=\frac12\sum_{(i,j)\in\mc E}\zeta(x_i,x_j)\,.$$ 
%
%
\end{example}

\begin{example}[Best-shot public good game]
	\label{example: public good}
	Consider a graph $\mc G$ and the game where every player $i$ in $\mc V$ has binary action set $\mc A_i=\{0,1\}$ and utility:
	$$
	u_i(x)=\left\{\ba{rcl}
	1-c &\se& x_i=1  \\
	1 &\se& x_i=0  \text{ and } x_j=1 \text{ for some } j \in \mc N_i \\
	0 &\se& x_i=0  \text{ and } x_j=0 \text{ for every } j \in \mc N_i\,. 
	\ea\right.$$
	The game $u$ constructed in this way is the so called ``public good game''.  It models a more complex behaviour for the population $\mc V$: players benefit form acquiring some good, represented by taking action $1$ and which is public in the sense that it can be lent from one player to another. Taking action $1$ has a cost $c$, so players would prefer that one of their neighbors takes that action, but taking the action and paying the cost is still the best choice if no one of their neighbors does. 
	The public good game is a graphical game on $\mc G$ but it is not pairwise-separable. 
%
\end{example}

%
%
%

\section{Separability and minimal representation of games}\label{sec:separable-games}

In this section, we introduce the notions of \emph{separable function} with respect to a H-graph and \emph{separable game} with respect to a FDH-graph. Then, we prove that every function $f:\mc X\to\R$ and every game $u$ in $\mc U$ admit, respectively, a minimal H-graph and a minimal FDH-graph with respect to which they are separable. 

Throughout the section, we shall assume to have fixed a finite player set $\mc V$, nonempty action sets $\mc A_i$ for every player $i$ in $\mc V$, and let $\mathcal{X}= \prod_{i \in \mathcal{V}} \mc A_i$. We recall that $\mc U$ stands for the set of games with player set $\mc V$ and strategy profile set $\mc X$.

\subsection{Separable functions}
We start by introducing the following notion of separability for a function defined on product spaces. 

\begin{definition}\label{def:separable-function}
A function\footnote{Every statement and reasoning in this subsection continues to hold true for function $f:\mc X\to\mc Z$, where $\mc Z$ is an arbitrary Abelian group.} $f:\mc X \to\R$ is \emph{$\mc H$-separable}, where $\mc H=(\mc V,\mc L)$ is a H-graph,  
	if there exist functions 
	$f_{\mc J}:\prod_{j \in \mc J} \mc A_j\to\R$, for $\mc J$ in $\mc L$, such that 
	\be\label{eq:separable-function}f(x) = \sum_{\mc J\in\mc L} f_{\mc J}(x_{\mc J})\,,\qquad \forall x\in\mc X\,.\ee
\end{definition}\medskip
Separability of a function $f:\mc X\to\R$ with respect to a H-graph $\mc H$ thus consists in decomposability of $f(x)$ as a sum of functions each depending exclusively on the variables $x_{\mc J}$ associated to an undirected hyperlink $\mc J$ of $\mc H$. 

%
%
Notice that every function $f:\mc X\to\R$ is $\mc H$-separable with respect to the trivial H-graph $\mc H=(\mc V,\{\mc V\})$ having a unique undirected hyperlink consisting of all nodes.  
Moreover, given two H-graphs $\mc H_1$ and $\mc H_2$ such that $\mc H_1\preceq\mc H_2$, we have that if $f$ is $\mc H_1$-separable, then it is also $\mc H_2$-separable. In particular, a function $f$ is $\mc H$-separable if and only if it is $\ov{\mc H}$-separable, where we recall that $\ov{\mc H}$ is the simple H-graph associated to $\mc H$. 

The following fundamental technical result will be instrumental to all our future derivations.
\begin{lemma}\label{lemma:separable-intersection}
Let a function $f:\mc X\to\R$ be both $\mc H_1$-separable and $\mc H_2$-separable for two H-graphs $\mc H_1$ and $\mc H_2$.  Then, $f$ is also $\mc H$-separable, where $\mc H=\mc H_1 \sqcap \mc H_2$.
\end{lemma}
\begin{proof} Let $\Sigma_f$ be the family of all H-graphs $\mc H$ such that $f$ is $\mc H$-separable and, for $i=1,2$, let $\mc H_i=(\mc V,\mc L_i)$ in $\Sigma_f$ be an H-graph such that $f$ is $\mc H_i$-separable. We can then write
\be\label{gJ0}f(x)=\sum_{\mc J\in\mc L_1}g^{(1)}_{\mc J}(x_{\mc J})=\sum_{\mc K\in\mc L_2}g^{(2)}_{\mc K}(x_{\mc K})\,,\qquad\forall x\in\mc X\,.\ee
Then, for every $\mc J$ in $\mc L_1$, we have that 
\be\label{gJ1}g^{(1)}_{\mc J}(x_{\mc J})=\sum_{\mc K\in\mc L_2}g^{(2)}_{\mc K}(x_{\mc K})-\sum_{\mc I\in\mc L_1\setminus\{\mc J\}}g^{(1)}_{\mc I}(x_{\mc I})\,,\qquad\forall x\in\mc X\,.\ee
Now, observe that, since the lefthand side of \eqref{gJ1} is independent from $x_{\mc V\setminus\mc J}$, so is its righthand side. Therefore, we may rewrite  \eqref{gJ1} as
\be\label{gJ2}g^{(1)}_{\mc J}(x_{\mc J})=\sum_{\mc K\in\mc L_2}h^{(2)}_{\mc K\cap\mc J}(x_{\mc K\cap\mc J})-\sum_{\mc I\in\mc L_1\setminus\{\mc J\}}h^{(1)}_{\mc I\cap\mc J}(x_{\mc I\cap\mc J})\,,\qquad\forall x\in\mc X\,,\ee
where, for an arbitrarily chosen $y$ in $\mc X$,  
\be\label{gJ3}h^{(i)}_{\mc K\cap\mc J}(x_{\mc K\cap\mc J})=g^{(i)}_{\mc K}(x_{\mc K\cap\mc J},y_{\mc K\setminus\mc J})\,,\qquad \forall i=1,2\,, \quad\forall\mc K\in\mc L_1\cup\mc L_2\,,\quad\forall x\in\mc X\,.\ee
It then follows from \eqref{gJ0}, \eqref{gJ2} and \eqref{gJ3} that 
\be\label{gJ4}f(x)=\sum_{\mc J\in\mc L_1}\sum_{\mc K\in\mc L_2}h^{(2)}_{\mc K\cap\mc J}(x_{\mc K\cap\mc J})-\sum_{\mc J\in\mc L_1}\sum_{\mc I\in\mc L_1\setminus\{\mc J\}}h^{(1)}_{\mc I\cap\mc J}(x_{\mc I\cap\mc J})\,,\qquad\forall x\in\mc X\,.\ee

Observe that \eqref{gJ4} is not yet the desired separability decomposition because of the presence of the second term in its righthand side. However, a suitable iterative application of \eqref{gJ4} allows us to prove the claim. To formally see this, it is convenient to first introduce the following definition. 
Given a H-graph $\mc H=(\mc V,\mc L)$, let the H-graphs $\sqcap^k\mc H=(\mc V, \mc L)$ be defined by
\be\label{intersect-distinct}
 \mc L=\{J_1\cap\cdots\cap J_k\;|\; J_s\in\mc L\;\forall s,\; J_s\neq J_t\;\forall s\neq t\}\ee
 and notice that $\sqcap^2( \sqcap^k\mc H)\preceq\sqcap^{k+1}\mc H$.
We can now interpret \eqref{gJ4} as saying that 
\be\label{induction0}(\mc H_1\sqcap\mc H_2)\sqcup  (\sqcap^2\mc H_1)\in\Sigma_f \,.\ee
We now prove by induction that, for every $k\ge2$, 
\be\label{induction1}(\mc H_1\sqcap\mc H_2)\sqcup  (\sqcap^k\mc H_1)\in\Sigma_f \,.\ee
Indeed, assume that \eqref{induction1} holds true for a certain $k$ and let us prove it for $k+1$. Considering that \eqref{induction0} is true for any pair of H-graphs $\mc H_1,\mc H_2$ in $\Sigma_f$, if we apply it replacing $\mc H_1$ with $(\mc H_1\sqcap\mc H_2)\sqcup  (\sqcap^k\mc H_1)$, we obtain that
\be\label{induction2}(((\mc H_1\sqcap\mc H_2)\sqcup  (\sqcap^k\mc H_1))\sqcap\mc H_2)\sqcup (\sqcap^2((\mc H_1\sqcap\mc H_2)\sqcup  (\sqcap^k\mc H_1)))\in\Sigma_f\,.
\ee
Notice now that 
\be\label{induction3}((\mc H_1\sqcap\mc H_2)\sqcup  (\sqcap^k\mc H_1))\sqcap\mc H_2\preceq \mc H_1\sqcap\mc H_2\,,\ee
\be\label{induction4}\sqcap^2((\mc H_1\sqcap\mc H_2)\sqcup  (\sqcap^k\mc H_1))\preceq (\mc H_1\sqcap\mc H_2)\sqcup (\sqcap^2(\sqcap^k\mc H_1))\preceq (\mc H_1\sqcap\mc H_2)\sqcup (\sqcap^{k+1}\mc H_1)\,.\ee
Relations \eqref{induction2}, \eqref{induction3}, and \eqref{induction4} imply \eqref{induction1} for $k+1$. Therefore, \eqref{induction1} holds true for every value of $k$. Finally, notice that, for $k> |\mc L|$, $\sqcap^k\mc H$ is the H-graph with an empty set of hyperlinks. This proves that $\mc H_1\sqcap\mc H_2\in\Sigma_f$.
%
%
%
%
\qed\end{proof}

Lemma \ref{lemma:separable-intersection} and the foregoing considerations  motivate the following definition. 

\begin{definition}\label{def:minimal-H-graph}
A H-graph $\mc H=(\mc V,\mc L)$ is a \emph{minimal H-graph} for a function $f:\mc X \to\R$ if
$\mc H$ is simple, $f$ is $\mc H$-separable, and $\mc H\preceq\tilde{\mc H}$ for every H-graph $\tilde{\mc H}=(\mc V,\tilde{\mc L})$ such that $f$ is $\tilde{\mc H}$-separable.  
	\end{definition}\medskip

We can now prove the following result. 

\begin{proposition}\label{proposition:minimal-separation}
Every function $f:\mc X\to\R$ admits a unique minimal H-graph $\mc H_f$. 
\end{proposition}
\begin{proof}
Lemma \ref{lemma:separable-intersection} implies that $f$ is $\mc H$-separable, where $\mc H$ is the $\sqcap$-intersection of all H-graphs $\tilde{\mc H}=(\mc V,\tilde{\mc L})$ such that $f$ is $\tilde{\mc H}$-separable. Then, $f$ is also $\ov{\mc H}$-separable. Now, let $\mc H_f=\ov{\mc H}$ and notice that $\mc H_f\preceq\mc H\preceq\tilde{\mc H}$ for every  H-graph $\tilde{\mc H}$ such that $f$ is $\tilde{\mc H}$-separable. Since by construction $\mc H_f$ is simple, we have that $\mc H_f$ is the minimal H-graph of $f$. 
\qed\end{proof}

\subsection{Separable games}\label{sec:separable-games-def}

We now introduce the following notion of separability for a game. 
\begin{definition}\label{def:separable-game}
Given a FDH-graph  $\mc F=(\mc V,\mc D)$, a game $u$ in $\mc U$ is \emph{$\mc F$-separable}  if 
	the utility of each player $i$ in $\mc V$ 
	can be decomposed as  
	\be\label{eq:utility-decomposition-separable}u_i(x)=\sum_{(i,\mc J)\in\mc D}u_{i}^{\mc J}(x_i,x_{\mc J})+n_i(x_{-i})\,,\ee
	where  $u_i^{\mc J}:\mc A_i\times\prod_{j\in\mc J}\mc A_j\to\R$ are functions that depend on the actions of player $i$ and of players in the subset $\mc J$ of head nodes of hyperlink $(i,\mc J)$ only, while $n_i:\mc X_{-i}\to\R$ is a non-strategic component that does not depend on the action of player $i$. 
\end{definition} \medskip

Definition \ref{def:separable-game} captures not only locality of the relative influences among the players in the game, but also the fact that players may have separate interactions with different groups of other players. Up to a non-strategic component, this grouping of the player set is modeled as a FDH-graph with node set coinciding with the player set $\mc V$ and where each group jointly influencing player $i$ corresponds to a directed hyperlink with tail node $i$. 

We now make a few technical remarks on this definition. In particular, we connect the current notion of separability for games to the one previously introduced for functions as well as to the concept of graphical games.

\begin{remark}By definition, the notion of separability introduced above is invariant with respect to strategic equivalence, i.e., a game $u$ is $\mc F$-separable if and only every game $\tilde u$ that is strategically equivalent to $u$ is $\mc F$-separable. In that, Definition \eqref{def:separable-game} differs from other notions proposed in the literature, see, e.g., that of ``graphical multi-hypermatrix game'' \cite{Ortiz.Irfan:2017}. 
In every class of strategically equivalent games there exists a unique game $\ov u$, called \emph{normalized}  \cite{Candogan.ea:2011}, that satisfies the following property
\be\label{eq:normalized}
\sum_{y\sim_i x} \ov u_i(y) = 0\,,\qquad \forall x \in \mathcal{X},\ \forall i\in\mc V\,.\ee
In particular, for every finite game $u$, the normalized strategically equivalent game $\ov u$ can be obtained by the formula 
%
\be\label{eq:u-normalized}
\ov u_i(x)=u_i(x)-\frac1{|\mc A_i|} \sum_{y\sim_i x} u_i(y)\,,\qquad \forall x \in \mathcal{X},\ \forall i\in\mc V\,.\ee
Notice that for normalized games, if formula \eqref{eq:utility-decomposition-separable} holds true, it must do so with $n_i(x_{-i})\equiv0$ for every player $i$ in $\mc V$. 
\end{remark}\medskip

\begin{remark}\label{rem:local} Separability of a game can be equivalently expressed in terms of separability of the single utility functions. Indeed, given an FDH-graph $\mc F=(\mc V,\mc D)$, consider, for each $i$ in $\mc V$, the `local' H-graph  
		\be\label{H-local}\mc H_i=(\mc V,\mc L_i),\quad \mc L_i=\{\{i\}\cup\mc J:\,(i,\mc J)\in\mc D\}\cup\{\mc V\setminus\{i\}\}\,.\ee
Notice that a game $u$ in $\mc U$ is $\mc F$-separable in the sense of Definition \ref{def:separable-game} if and only if, for every $i$ in $\mc V$, the utility function $u_i$ is $\mc H_i$-separable in the sense of Definition \ref{def:separable-function}.
\end{remark}\medskip

\begin{remark}\label{rem:graph-sep} A direct consequence of Definitions \ref{def:graphical-game} and \ref{def:separable-game} and  of relations \eqref{FG} and \eqref{GF} is that, given a FDH-graph $\mc F$, every $\mc F$-separable game $u$ is graphical with respect to the graph $\mc G^{\mc F}$. Similarly, if $u$ is graphical with respect to $\mc G$, it is $\mc F^{\mc G}$-separable.
\end{remark}\medskip


A game $u$
can be $\mc F$-separable with respect to different FDH-graphs $\mc F$. In fact, if a game is $\mc F_1$-separable for a given FDH-graph $\mc F_1$, it is also ${\mc F_2}$-separable for every FDH-graph $\mc F_2$ such that $\mc F_1\preceq\mc F_2$. 
Hence, in particular, a game $u$ is $\mc F$-separable if and only if it is $\ov{\mc F}$-separable. 

A natural question, addressed below, is whether there exists a FDH-graph $\mc F$ that captures the minimal structure of a game. Such minimality property is formalized by the following definition. 

\begin{definition}\label{def:minimalFDH-graph}
	A FDH-graph $\mc F=(\mc V,\mc D)$ is the \emph{minimal } FDH-graph of a game $u$ in $\mc U$ if $\mc F$ is simple, $u$ is $\mc F$-separable, and $\mc F\preceq\tilde{\mc F}$ for every FDH-graph $\tilde{\mc F}=(\mc V,\tilde{\mc D})$ such that $u$ is $\tilde{\mc F}$-separable. 
\end{definition} \medskip

The following result states that every game admits a minimal  FDH-graph with respect to which it is separable.
\begin{theorem}\label{theo:minimal-class}
	Every game $u$ in $\mc U$ admits a unique minimal FDH-graph $\mc F_u=(\mc V,\mc D)$. 
\end{theorem}
\begin{proof}
Let $u$ in $\mc U$ be a game that is separable with respect to two FDH-graphs $\mc F_1=(\mc V,\mc D_1)$ and $\mc F_2=(\mc V, \mc D_2)$. For every player $i$ in $\mc V$, consider the corresponding local H-graphs $\mc H^1_i$ and $\mc H^2_i$ as defined in \eqref{H-local}. Following Remark \ref{rem:local} we have that the utility function $u_i$ is $\mc H^s_i$-separable for $s=1,2$ and, consequently, because of Lemma \ref{lemma:separable-intersection}, also $\mc H^1_i\sqcap \mc H^2_i$-separable.
Since, for $i$ in $\mc V$,  $\mc H^1_i\sqcap \mc H^2_i$ are the local H-graphs associated with the intersection $\mc F_1 \sqcap \mc F_2$, using again Remark \ref{rem:local} we deduce that $u$ is $\mc F_1 \sqcap \mc F_2$-separable.
Then, the result follows arguing in the same way as in the proof of Proposition \ref{proposition:minimal-separation}. 
\qed\end{proof} 

The minimal separability and graphicality properties of a game are connected.
More specifically, the relation between the minimal FDH-graph $\mc F_u$ of a game $u$ and its minimal graph $\mc G_u$ is clarified in the following result.

\begin{corollary}\label{cor:minimal1} For every game $u$ in $\mc U$, the minimal graph $\mc G_u$ and the minimal FDH-graph $\mc F_u$ are related by $\mc G_u=\mc G^{\mc F_u}$.
\end{corollary}
\begin{proof} 
We know from Remark \ref{rem:graph-sep} that $u$ is a $\mc F^{\mc G_u}$-separable $\mc G^{\mc F_u}$-game. From this and the relations \eqref{GFrelation} we have that 
$$\mc G_u\subseteq \mc G^{\mc F_u},\quad \mc F_u\preceq \mc F^{\mc G_u}\;\Rightarrow\; \mc G^{\mc F_u}\subseteq \mc G^{\left(\mc F^{\mc G_u}\right)}=\mc G_u \,.$$
This concludes the proof.\qed
\end{proof}

\subsection{Examples} 
As discussed earlier, pairwise-separable network games with respect to a graph $\mc G$ are $\mc F$-separable with respect to the FDH-graph $\mc F$ obtained interpreting $\mc G$ as a FDH-graph. Below we analyze two particular graphical games that are not pairwise-separable, but rather separable with respect to a different FDH-graph.

\begin{example}[Two-level coordination game]\label{example:two-level coord} We consider the following variation of the network coordination game presented in Example \ref{example:coord.game}. We fix a set of players $\mc V$ and the same action set for all players: $\mc A=\mc A_i=\{0,1\}$ for all $i$ in $\mc V$. For every pair of nodes $i,j$ we consider functions $u_{ij}:\mc A^2\to\R$ defined as the pairwise utility function $\zeta$ of the network coordination game \eqref{eq:coord.game}. 
We now consider an undirected graph $\mc G=(\mc V,\mc E)$ and, for every $i$ in $\mc V$, functions $\tilde u_i:\mc A^{\mc N^{\bullet}_i}\to\R$ given by
$$\tilde u_i(x)=\left\{\begin{array}{ll} L \quad&{\rm if}\; x_i=x_k\, \forall k\in \mc N_i\\
0 \quad&{\rm otherwise}\end{array}\right.$$
where $L>0$. We finally define the utility of player $i$ as:
$$u_i(x)=\sum\limits_{j\neq i}u_{ij}(x_i, x_j)+\tilde u_i(x_{\mc N^{\bullet}_i})$$
The interpretation is the following: each agent has a benefit that is in part linearly proportional to the number of individuals playing the same action and, additionally, it has an extra value $L$ if the agent's action is in complete agreement with her neighbors. This type of utility function models, for example, the situation where player's action represents the acquisition of a new technology and the benefit to a player comes from two channels: the range of diffusion of the technology in the whole population and the opportunity to use such technology with her strict collaborators.
If we consider the FDH-graph $\mc F=(\mc V, \mc D)$ where
$$\mc D=\{(i, \{j\})\,,\, j\neq i\}\cup\{(i, \mc N_i),\; i\in\mc V\}$$
we have that the game $u$ is $\mc F$-separable. Notice that this is not the minimal FDH-graph for $u$ as the minimal one is
$\mc F_u=(\mc V, \mc D_u)=\ov{\mc F}$ where
$$\mc D_u=\{(i, \{j\})\,,\, j\not \in\mc N^{\bullet}_i\}\cup\{(i, \mc N_i),\; i\in\mc V\}\,.$$
\end{example}

\section{On the structure of potential games}\label{sec:potential}
In this section we focus on potential games and study how their separability properties are intertwined with the separability of the corresponding potential functions. This is the content of out next result Theorem \ref{theorem:potential}. We then derive, as a corollary, results on graphical potential games first appeared in \cite{Babichenko.Tamuz:2016} and we provide an alternative proof of the Hammersley-Clifford theorem for Markov random fields.
%
%
%

\begin{theorem}\label{theorem:potential}
Let $u$ in $\mc U$ be a potential game with potential function $\phi$. Then, the minimal FDH-graph of $u$ is the undirected FDH-graph associated to the minimal H-graph of $\phi$, i.e., 
\be\label{potential-separable}\mc F_u=\mc F^{\mc H_{\phi}}\ee
\end{theorem}
\begin{proof}
As discussed in Section \ref{sec:games}, every potential game  $u$ in $\mc U$ is strategically equivalent to a game $u^{\phi}$ whose players have utilities  all equal to the potential function $\phi$. Since $\phi$ is $\mc H^{\phi}$-separable, with $\mc H^{\phi} = (\mc V, \mc L_{\phi})$, we can write
%
%
\be\label{uiphi}u_i(x)=\phi(x)+n_i(x_{-i})=\sum_{\mc K\in\mc L_{\phi}}\phi_{\mc K}(x_{\mc K})+n_i(x_{-i})\,,\ee
		for some function $n_i:\mc X_{-i}\to\R$.  This shows that $u$ is $\mc F^{\mc H_{\phi}}$-separable, therefore $\mc F_u\preceq\mc F^{\mc H_{\phi}}$.
		
Consider now the local H-graphs $\mc H_i=(\mc V,\mc L_i)$, for $i$ in $\mc V$, associated to the minimal FDH-graph $\mc F_u$ in the sense of \eqref{H-local}. By Remark \ref{rem:local}, every utility function $u_i$ is $\mc H_i$-separable and thus from the first equality in \eqref{uiphi} we get that also $\phi$ is $\mc H_i$-separable, for every $i$ in $\mc V$. By definition, every hyperlink $(i,\mc J)$ of the FDH-graph $\mc F^{\mc H_{\phi}}$ is such that $\{i\}\cup \mc J$ in $\mc L_{\phi}$ is an undirected hyperlink of $\mc H_{\phi}$. Since the potential function $\phi$ is $\mc H_i$-separable and $\mc H_{\phi}$ is the minimal H-graph of $\phi$, this implies that there exists $\mc K$ in $\mc L_i$ such that $\{i\}\cup \mc J\subseteq\mc K$. By the way the local H-graph $\mc H_i$ is defined, it follows that necessarily $(i,\mc K\setminus\{i\})$ in $\mc D_u$ is a hyperlink of the FDH-graph $\mc F_u$. We have thus just proved that $\mc F^{\mc H_{\phi}}\preceq\mc F_u$. The claim then follows as we had already shown that $\mc F_u\preceq\mc F^{\mc H_{\phi}}$. 
%
\qed\end{proof}

The above result implies the following relation between the minimal graph of a potential game and the separability of the potential function. Given an undirected graph $\mc G$ we denote by $\mc Cl(\mc G)$ the set of maximal cliques in $\mc G$, and let $\mc H^{\mc Cl}_{\mc G}=(\mc V,\mc Cl(\mc G))$ be the cliques H-graph of $\mc G$.  

\begin{corollary}\label{cor:potential} Let $u$ in $\mc U$ be a potential game with potential function $\phi$. Then, the minimal graph $\mc G_u$ associated with $u$ is undirected. Moreover, $u$ is a $\mc G$-game for an undirected graph $\mc G$,  if and only if its potential function $\phi$ is $\mc H^{\mc Cl}_{\mc G}$-separable.
\end{corollary}
\begin{proof} Consider the minimal FDH-graph $\mc F_u=(\mc V, \mc D_{u})$ of $u$
and the minimal H-graph $\mc H_\phi=(\mc V, \mc L_{\phi})$ of $\phi$.
It follows from Corollary \ref{cor:minimal1} and relation \eqref{potential-separable} that
$\mc G_u=(\mc V, \mc E_u)=\mc G^{\mc F_u}$ is the graph associated to the undirected FDH-graph $\mc F_u=\mc F^{\mc H_{\phi}}$ and thus it is itself undirected. 

Suppose now that $u$ is a $\mc G$-game for some undirected graph $\mc G=(\mc V,\mc E)$. Let $\mc K$ in $\mc L_{\phi}$. Then,  for every $i$ in $\mc K$,   \eqref{potential-separable} implies that $(i,\mc K\setminus\{i\})$ belongs to $\mc D_u$ and thus $(i,j)$ belongs to $\mc E_u$ for every $i\ne j$ in $\mc K$. This says that $\mc K$ is a clique in $\mc G_u$ and thus also in $\mc G$, since $\mc G_u\preceq\mc G$. Therefore, $\mc H_{\phi}\preceq \mc H^{\mc Cl}_{\mc G}$, thus showing that $\phi$ is $\mc H^{\mc Cl}_{\mc G}$-separable. 

Conversely, if $\phi$ is $\mc H^{\mc Cl}_{\mc G}$-separable, then necessarily $\mc H_{\phi}\preceq \mc H^{\mc Cl}_{\mc G}$, so that that every undirected hyperlink in $\mc L_{\phi}$ is contained in a clique of $\mc G$. By Corollary \ref{cor:minimal1}, the minimal graph $\mc G_u$ is the graph associated with the FDH-graph $\mc F_u$ and, by Theorem \ref{theorem:potential}, $\mc F_u=\mc F^{\mc H_{\phi}}$ is undirected. It then follows that $\mc G_u = \mc G^{\mc F^{\mc H_{\phi}}} \preceq \mc G^{\mc F^{\mc H^{\mc Cl}_{\mc G}}} = \mc G$, thus showing that $u$ is a $\mc G$-game. \qed
%
\end{proof}

The second part of Corollary \ref{cor:potential} is equivalent to Theorems 4.2 and 4.4 in \cite{Babichenko.Tamuz:2016}. In this paper, the authors prove their results relying on the Hammersley-Clifford theorem. Our proofs are instead self-contained and in the next subsection we actually show that 
the Hammersley-Clifford theorem can be derived from our results. 

In fact, we wish to emphasize that Theorem \ref{theorem:potential} is more informative than Corollary \ref{cor:potential}. Indeed, the latter does not relate the the minimal separability of a potential game with that of its corresponding potential function. This is evident, e.g., in the special case of a potential game $u$ that is pairwise separable with respect to an undirected graph $\mc G=(\mc V, \mc E)$. In this case,  \eqref{potential-separable} implies that the potential function $\phi$ is separable with respect to the H-graph $\mc H$ coinciding with $\mc G$, i.e., 
that it can be decomposed in a pairwise fashion
$$\phi(x)=\sum\limits_{(i,j)\in\mc E}\phi_{ij}(x_i, x_j)\,,$$
for some symmetric functions $\phi_{ij}(x_i, x_j)=\phi_{ji}(x_j, x_i)$.
This is in general a much finer decomposition than the one on the maximal cliques of $\mc G$.


\subsection{Markov random fields and the Hammersley-Clifford Theorem}
In this subsection we show how the celebrated Hammersley-Clifford Theorem on the structure of Markov random fields can be deduced from Corollary \ref{cor:potential}. 

Consider an undirected graph $\mc G=(\mc V,\mc E)$ and a vector of finite-valued random variables $X=(X_i)_{i\in\mc V}$ indexed by the nodes of $\mc G$. Denote by $\mc A_i$ the set where the random variable $X_i$ takes its values and put $\mc X=\prod_{i}\mc A_i$.
For every subset $\mc W\subseteq \mc V$, let $X_{\mc W}$ denote the subvector of $X$ consisting of the random variables $X_i$ with $i$ in $\mc W$. We shall refer to the random vector $X$ as positive if its probability distribution is equivalent to the product of the marginals, namely, if $\P(X=x)>0$ whenever $\P(X_i=x_i)>0$ for every $i$ in $\mc V$.


 We shall refer to the random vector $X$ as a Markov random field (with respect to $\mc G$) if, for every node $i$ in $\mc V$, $X_i$ and $X_{\mc V\setminus\mc N_i^\bullet}$ are conditionally independent given $X_{\mc N_i}$.\footnote{In the literature on probabilistic graphical models, this is referred to as the \emph{local} Markov property \cite[Ch.~3.1]{Lauritzen:1996}, which is known to be implied by the so-called \emph{global} Markov property and in turn to imply the so-called \emph{pairwise} Markov property. \cite[Proposition 3.4]{Lauritzen:1996}. The three Markov properties are in fact known to be all  equivalent to one another for positive random vectors \cite[Theorem 3.7]{Lauritzen:1996}.} 

\begin{theorem} Let $X$ be a positive Markov random field with respect to an undirected graph $\mc G=(\mc V,\mc E)$. Then, its probability distribution admits the following decomposition:
\be\P(X=x)=\prod\limits_{\mc C\in\mc Cl(\mc G)}\zeta_{\mc C}(x_{\mc C}),\quad\forall x\in\mc X\,,
\ee
where $\mc Cl(\mc G)$ is the family of maximal cliques of the graph $\mc G$.
\end{theorem}
\begin{proof} Without loss of generality we can assume that $\P(X_i=x_i)>0$ for every $i$ in $\mc V$ and $x_i$ in $\mc A_i$ so that, by the positivity assumption we have that $\P(X=x)>0$ for every $x$ in $\mc X$. Let \be\label{phi=logP}\phi(x)=\log \P(X=x)\,,\qquad\forall x\in\mc X\ee and consider the potential game $u=u^{\phi}$ in $\mc U$ with utility functions $u_i(x)=\phi(x)$ for every $i$ in $\mc V$. 

We shall now prove that $u$ is ${\mc G}$-graphical. Indeed, conditional independence implies that 
$$\begin{array}{rcl}\P(X=x)&=&\P(X_{\mc N_i}=x_{\mc N_i})\P(X_{\mc V\setminus \mc N_i}
=x_{\mc V\setminus \mc N_i}\,|\, X_{\mc N_i}=x_{\mc N_i})\\[7pt]
&=&\P(X_{\mc N_i}=x_{\mc N_i})\P(X_{i}=x_{i} \,|\, X_{\mc N_i}=x_{\mc N_i})\P(X_{\mc V\setminus\mc N^{\bullet}_i}=x_{\mc V\setminus\mc N^{\bullet}_i} \,|\, X_{\mc N_i}=x_{\mc N_i})\\[7pt]
&=&\P(X_{\mc N_i^{\bullet}}=x_{\mc N_i^{\bullet}})\P(X_{\mc V\setminus\mc N^{\bullet}_i}=x_{\mc V\setminus\mc N^{\bullet}_i} \,|\, X_{\mc N_i}=x_{\mc N_i})
\,,
\end{array}$$
 for every $i$ in $ \mc V$.
We can then write $u_i(x)=u_i^{\mc N_i}(x_i, x_{\mc N_i})+n_i(x_{-i})$ where $$u_i^{\mc N_i}(x_i, x_{\mc N_i})=u_i^{\mc N_i}(x_{\mc N_i^{\bullet}})=\log \P(X_{\mc N_i^{\bullet}}=x_{\mc N_i^{\bullet}})\, 
$$ only depends on the actions played by player $i$ and her neighbors in $\mc N_i$, while  
$$n_i(x_{-i})=\log \P(X_{\mc V\setminus\mc N^{\bullet}_i}=x_{\mc V\setminus\mc N^{\bullet}_i} \,|\, X_{\mc N_i}=x_{\mc N_i})$$ is a non strategic term. Thus $u$ is $\mc G$-graphical so that Corollary \ref{cor:potential} implies that  its potential function $\phi$ is $\mc H^{\mc Cl}_{\mc G}$-separable. Together with \eqref{phi=logP}, this yields the claim.
\qed\end{proof}

%

%

\section{The potential and harmonic components of a separable game}\label{sec:main-results2}

In this section, we finally consider a recent result presented in \cite{Candogan.ea:2011} regarding the decomposition of a general game and we investigate how the concept of separability interacts with this decomposition. 

We first introduce another game theoretic notion. 
A game $u $ in $ \mc U$ is referred to as \emph{harmonic} \cite{Candogan.ea:2011} if 	
\be\label{eq:harmonic}
\sum_{i \in \mathcal{V}} \sum_{y\sim_i x} [u_i(x) - u_i(y)] = 0\,, \ee
for every strategy profile $x$ in $\mc X$. 
In \cite[Theorem 4.1]{Candogan.ea:2011} it is shown that every game $u$ can be decomposed as a sum of three games
\be u=u_{\rm pot}+u_{\rm har}+n\ee
where  $u_{\rm pot}$ is a potential game, $u_{\rm har}$ is a harmonic game, and $n$ is a non-strategic game. This decomposition is unique up to non-strategic components in $u_{\rm pot}$ and $u_{\rm har}$. In particular, it is unique if we assume that $u_{\rm pot}$ and $u_{\rm har}$ are both normalized. 

The following result shows that in general the potential and  harmonic components of an $\mc F$-separable game are separable on the underlying undirected FDH-graph $\mc F^{\leftrightarrow}$, whose hyperlink set can be obtained from $\mc F$'s as specified in \eqref{underlying-def}. 

\begin{theorem}\label{theo:main}
	Let $u$ in $\mc U$ be a finite game that is $\mc F$-separable with respect to a FDH-graph $\mc F=(\mc V,\mc D)$. Then, $u_{\rm pot}$ and $u_{\rm har}$ are $\mc F^{\leftrightarrow}$-separable.

\end{theorem}

\begin{proof}
We consider the decomposition \eqref{eq:utility-decomposition-separable} and for each $(i,\mc J)$ in $\mc D$ we define the auxiliary game
\be\label{eq:uih}
u^{(i,\mc J)}_i(x)=u_i^{\mc J}(x_i,x_{\mc J})\,,\qquad u^{(i,\mc J)}_j(x)=0\,,\qquad \forall j\in\mc V\setminus\{i\} \,.
\ee
This is a game where all players have zero utility except for player $i$, who has utility equal to the corresponding term $u_i^{\mc J}(x_i,x_{\mc J})$ in \eqref{eq:utility-decomposition-separable}.
We can write
\be\label{eq:A}u=\sum_{(i,\mc J)\in\mc D} u^{(i,\mc J)}\,,\ee
For each game $u^{(i,\mc J)}$ we now consider its restriction $\hat u$ to the set of players $\hat{\mc V}=\{i\}\cup\mc J$ and strategy profile set $\hat{\mc X}=\prod_{h\in \hat{\mc V}}\mc A_h$. Formally,
$$\hat u_h(x)=\left\{\begin{array}{ll} u_i^{\mc J}(x_i,x_{\mc J})\quad & {\rm if}\, h=i\\ 0\quad & {\rm otherwise} \,.\end{array}\right.$$
Consider now the potential and the harmonic part of $\hat u$, respectively, $\hat u_{\rm pot}$ and $\hat u_{\rm har}$ and extend them to $\mc V$ as games $u^{(i,\mc J)}_{\rm pot}$ and $u^{(i,\mc J)}_{\rm har}$ such that
$$(u^{(i,\mc J)}_{\rm pot})_h=(u^{(i,\mc J)}_{\rm har})_h=0\quad\forall h\not\in \{i\}\cup\mc J \,.$$
A direct check shows that $u^{(i,\mc J)}_{\rm pot}$ is potential, $u^{(i,\mc J)}_{\rm har}$ is harmonic and $u^{(i,\mc J)}=u^{(i,\mc J)}_{\rm pot}+u^{(i,\mc J)}_{\rm har}$ up to non-strategic terms. Then, up to non strategic terms,
$$u_{\rm pot}=\sum_{(i,\mc J)\in\mc D_u}u^{(i,\mc J)}_{\rm pot},\quad u_{\rm har}=\sum_{(i,\mc J)\in\mc D_u}u^{(i,\mc J)}_{\rm har}$$
Notice, finally, that since $u^{(i,\mc J)}_{\rm pot}$ and $u^{(i,\mc J)}_{\rm har}$ are separable with respect to the FDH-graph whose directed hyperlinks are of type $(h,\mc K)$ with $h$ in $\{i\}\cup \mc J$ and $\{h\}\cup\mc K=\{i\}\cup \mc J$, we have that 
$u^{(i,\mc J)}_{\rm pot}$ and $u^{(i,\mc J)}_{\rm har}$ are $\mc F^{\leftrightarrow}$-separable. This proves the result. \qed
\end{proof}

We conclude this section by studying two examples that show how for  a non-potential game $u$, its potential and harmonic components may display additional interdependancies among the players, as captured by the underlying undirected FDH-graph $\mc F_u^{\leftrightarrow}\succeq\mc F_u$. 

\begin{example}
\label{example: decomposition best shot}\begin{figure}
		\centering
		\includegraphics[width=0.3\linewidth]{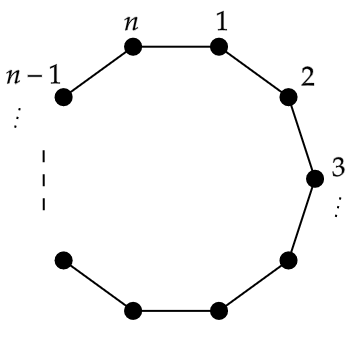}
		\caption{Ring graph of Example \ref{example: decomposition best shot} and \ref{example: decomposition strong coord}.}
		\label{fig: example dec best shot}
	\end{figure}
Consider the \emph{best-shot public good game} as introduced in Example \ref{example: public good}, with $n = |\mc V|$ players and defined on an undirected ring graph $\mc G= (\mc V, \mc E)$ with $n\ge3$ nodes, as shown in Figure \ref{fig: example dec best shot}. Thanks to the symmetric structure of the ring graph, we can express the utility function of player $i$ in $\mc V$ as
	\begin{equation}\label{eq: public good cycle}
	u_i(x) = \max\{x_{i-1},x_i,x_{i+1}\} - cx_i
	\end{equation}
	where the algebra on the indices is intended modulo $n$ and where $0\leq c \leq 1$ is the cost parameter.
	It is clear from the form of the utility in \eqref{eq: public good cycle} that the game $u$ is separable on the FDH-graph $\mc F = (\mc V, \mc D)$ with set of directed hyperlinks 
	$$
	\mc D = \left\lbrace (i, \{i-1, i+1\}), \,i \in \mc V \right\rbrace \,.
	$$
	In fact, as it is not possible to rewrite the $\max$ term in \eqref{eq: public good cycle} as the sum of three terms depending on $(x_{i-1},x_i)$, $(x_{i},x_{i+1})$, and $(x_{i-1},x_{i+1})$, respectively, the above can also be verified to be the minimal FDH-graph of $u$. 
	The utility functions of the normalized potential and harmonic components are given by
	\begin{equation}\label{equation: potential comp BS}
	u_{{\rm pot}, i}(x) = \left( |x_{i+1}-x_{i+2}| + |x_{i-1} - x_{i-2} |+ 4(x_{i+1}+x_{i-1}) - 2x_{i+1}x_{i-1} - 6(1-c) \right) \frac{1-2x_i}{12} \,,
	\end{equation}
	and, respectively, 
	\begin{equation}\label{equation: harmonc comp BS}
	u_{{\rm har}, i}(x) = \left( |x_{i+1}-x_{i+2}| + |x_{i-1} - x_{i-2}| - 2(x_{i+1}+x_{i-1})\right) \frac{1-2x_i}{12} \,,
	\end{equation}
for every player $i$ in $\mc V$.
	Notice that the harmonic component is independent from $c$, while the potential component contains an additive term that is linear on $c$ and that depends only on the action of player $i$ itself.
	In accordance with the decomposition result of Theorem \ref{theo:main}, equations \eqref{equation: potential comp BS} and \eqref{equation: harmonc comp BS} show that the normalized potential and harmonic components of $u$ are separable on the FDH-graph $\mc F^{\leftrightarrow} = (\mc V, \mc D^{\leftrightarrow})$ where
	$$
	\mc D^{\leftrightarrow} = \left\lbrace (i, \{i-2, i-1\}), (i, \{i-1, i+1\}), (i, \{i+1, i+2\}), \; i \in \mc V \right\rbrace \,.
	$$
	Accordingly, the potential function $\phi$ of $u_{\rm pot}$ is separable on the H-graph $\mc H^{\mc F^{\leftrightarrow}} = \mc H^{\mc F} = (\mc V, \mc L)$ with set of undirected hyperlinks
	$$
	\mc L = \left\lbrace \{ i-1, i, i+1\}, \, i \in \mc V \right\rbrace \,,
	$$
that  is displayed in Figure \ref{fig: hyper best shot}.
	\begin{figure}
		\centering
		\includegraphics[width=0.3\linewidth]{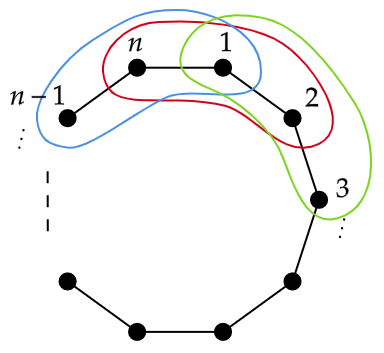}
		\caption{Hypergraph of the potential $\phi$ for Example \ref{example: decomposition best shot}.}
		\label{fig: hyper best shot}
	\end{figure}
	
\end{example}

\begin{example}\label{example: decomposition strong coord}
 For a graph $\mc G = (\mc V, \mc E)$, 	consider the game $u$ with player set $\mc V$, action set $\mc A_i=\{0,1\}$, and utilities	
	\begin{equation}
	u_i(x) = \begin{cases}
	1 \quad \text{if } x_i=x_j, \, \forall j \in \mc N_i \\
	0 \quad \text{otherwise}\,,
	\end{cases}
	\end{equation}
	for every player $i$ in $\mc V$. Such game coincide with the local interaction term of the two-level coordination game introduced in Example \ref{example:two-level coord} for the case $L=1$.
	Notice that, for a general graph $\mc G$, neither $u$ is a potential game nor it is pairwise separable on $\mc G$. We shall now study the structure of the game $u$ for two different graph topologies. 
	
When $\mc G$ is an undirected ring graph with $n$ nodes, as in Figure \ref{fig: example dec best shot}, then $u$ is  a potential game  and it is strategically equivalent to a multiple of the network coordination game on $\mc G$ introduced in Example \ref{example:coord.game} with $\zeta(x_i,x_j)= (-1)^{x_i-x_j}$, where the pairwise utilities are scaled by $\frac{1}{4}$. So, the normalized potential component $u_{\rm pot}$ is a network coordination game on $\mc G$, and it is pairwise separable on $\mc G$.
	
	
	\begin{figure}
		\centering
		\includegraphics[width=0.5\linewidth]{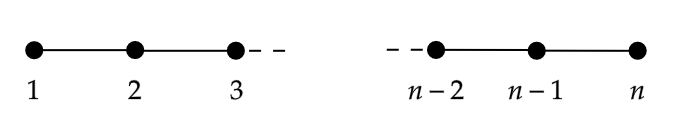}
		\caption{Line graph $L_n$ for Example \ref{example: decomposition strong coord}.}
		\label{fig:line}
	\end{figure}

	When $\mc G = L_n$ is an undirected line graph with $n $ nodes, as displayed in Figure \ref{fig:line}, then $u$ is not a potential game.
	The normalized potential component $u_{\rm pot}$ is a  weighted version of the network coordination game on the line $L_n$, as represented in Figure \ref{fig: weighted line}, where the pairwise utility function $\zeta(x_i,x_j)= (-1)^{x_i-x_j}$ is multiplied by different factors depending on the link $\{i,j\}$. 
	
	\begin{figure}
		\centering
		\includegraphics[width=0.5\linewidth]{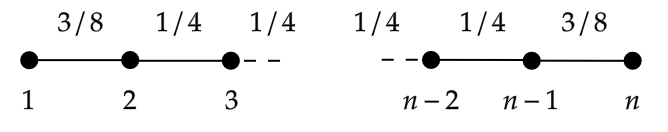}
		\caption{Line graph for the potential component of Example \ref{example: decomposition strong coord}. For each link $\{i,j\}$ the factor multiplying the pairwise utility $\zeta(x_i,x_j)$ is specified.}
		\label{fig: weighted line}
	\end{figure}
	
	In particular,
	\begin{itemize}
		\item for the internal links, joining nodes from $2$ to $n-1$, the pairwise utility $\zeta$ is scaled by $\frac{1}{4}$,
		\item for the extremal links $\{1,2\}$ and $\{n-1,n\}$, the pairwise utility $\zeta$ is scaled by $\frac{3}{8}$.
	\end{itemize}

	\begin{figure}
		\centering
		\includegraphics[width=0.5\linewidth]{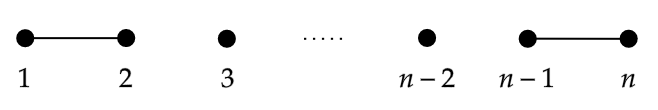}
		\caption{Graph for the harmonic component of Example \ref{example: decomposition strong coord}.}
		\label{fig: sub line}
	\end{figure}

	The harmonic component is pairwise separable on a graph with nodes $\mc V$ and with only two undirected links, namely $\{1,2\}$ and $\{n-1,n\}$, which is shown in Figure \ref{fig: sub line}. On such links two so-called matching pennies game take place, where nodes $1$ and $n$ aim to coordinate with their neighbors, while $2$ and $n-1$ aim to anti-coordinate. 
\end{example}

\section{Conclusion}\label{sec:conclusion} In this paper, we have introduced the notion of separability of a game with respect to a FDH-graph, encompassing and refining the notion of graphical game. Our proposed definition characterizes, up to a nonstrategic component, the way the utility functions can be split as the sum of functions depending on subgroups of players. We have shown that every game admits a minimal FDH-graph with respect to which it is separable. In the special case of potential games, we have shown that such a minimal FDH-graph enjoys a symmetry property and can in fact be identified with the undirected H-graph that describes the minimal separability structure of the potential function. This result  generalizes and refines one recently proved in \cite{Babichenko.Tamuz:2016} and, interestingly, it implies the celebrated Hammersley-Clifford Theorem for Markov random fields. Finally, we have investigated how the notion of separability for a game interacts with the potential-harmonic decomposition in \cite{Candogan.ea:2011} showing that in general additional interdependancies may emerge for the potential and harmonic components of a non-potential game. 

We believe that the potential implications of these fundamental results should be further investigated in a variety of directions among which: (i) studying how information on the minimal separating FDH-graph of a game can be used to lower the implementation complexity of learning algorithms; and (ii) understanding the implications of separability in the behavior of evolutionary and learning dynamics.

\section*{Acknowledgments.}
This work was partially supported by a MIUR grant ``Dipartimenti di Eccellenza 2018--2022'' [CUP: E11G18000350001], a MIUR  Research Project PRIN 2017 ``Advanced Network Control of Future Smart Grids'' (http://vectors.dieti.unina.it), the Swedish Research Council [2015-04066], and by the Compagnia di San Paolo. 

\bibliography{references}

\begin{thebibliography}{26}
\providecommand{\natexlab}[1]{#1}
\providecommand{\url}[1]{\texttt{#1}}
\expandafter\ifx\csname urlstyle\endcsname\relax
  \providecommand{\doi}[1]{doi: #1}\else
  \providecommand{\doi}{doi: \begingroup \urlstyle{rm}\Url}\fi

\bibitem[Morris(2000)]{Morris:2000}
S.~Morris.
\newblock Contagion.
\newblock \emph{The Review of Economic Studies}, 67\penalty0 (1):\penalty0
  57--78, 2000.

\bibitem[Jackson(2008)]{Jackson:2008}
M.O. Jackson.
\newblock \emph{Social and Economic Networks}.
\newblock Princeton University Press, 2008.

\bibitem[Easley and Kleinberg(2010)]{Easley.Kleinberg:2010}
D.~Easley and J.~Kleinberg.
\newblock \emph{Networks, Crowds, and Markets: Reasoning About a Highly
  Connected World}.
\newblock Cambridge University Press, 2010.

\bibitem[Jackson and Zenou(2015)]{Jackson.Zenou:2015}
M.~O. Jackson and Y.~Zenou.
\newblock \emph{Handbook of game theory with economic applications}, volume~4,
  chapter Games on networks, pages 95--163.
\newblock Elsevier, 2015.

\bibitem[Bramoull\'e and Kranton(2016)]{Bramoulle.Kranton:2016}
Yann Bramoull\'e and Rachel Kranton.
\newblock \emph{The Oxford Handbook of the Economics of Networks}, chapter
  Games played on networks.
\newblock Oxford University Press, 2016.

\bibitem[Ballester et~al.(2006)Ballester, Calv{\'o}-Armengol, and
  Zenou]{Ballester.ea:2006}
C.~Ballester, A.~Calv{\'o}-Armengol, and Y.~Zenou.
\newblock Who's who in networks. wanted: The key player.
\newblock \emph{Econometrica}, 74\penalty0 (5):\penalty0 1403--1417, 2006.

\bibitem[Bramoull\'e(2007)]{Bramoulle:2007}
Yann Bramoull\'e.
\newblock Anti-coordination and social interactions.
\newblock \emph{Games and Economic Behavior}, 2007.

\bibitem[Bramoull\'e and Kranton(2007)]{Bramoulle.Kranton:2007}
Yann Bramoull\'e and Rachel Kranton.
\newblock Public goods in networks.
\newblock \emph{Journal of Economic Theory}, 135:\penalty0 478--494, 2007.

\bibitem[Young(2009)]{Young:2009}
H.~Peyton Young.
\newblock \emph{The Economy As an Evolving Complex System, III: Current
  Perspectives and Future Directions}, chapter The diffusion of innovations in
  social networks.
\newblock Oxford University Press, 2009.

\bibitem[Montanari and Saberi(2010)]{Montanari.Saberi:2010}
Andrea Montanari and Amin Saberi.
\newblock The spread of innovations in social networks.
\newblock \emph{PNAS}, 2010.

\bibitem[Acemo\v{g}lu et~al.(2015)Acemo\v{g}lu, Ozdaglar, and
  Tahbaz-Salehi]{Acemoglu.ea:2015}
D.~Acemo\v{g}lu, A.~Ozdaglar, and A.~Tahbaz-Salehi.
\newblock \emph{The Oxford Handbook of the Economics of Networks}, chapter
  Networks, shocks, and systemic risk.
\newblock Oxford University Press, 2015.

\bibitem[Allouch(2015)]{Allouch:2015}
N.~Allouch.
\newblock On the private provision of public goods on networks.
\newblock \emph{Journal of Economic Theory}, 157:\penalty0 527--552, 2015.

\bibitem[Kearns et~al.(2001)Kearns, Littman, and Singh]{Kearns.ea:2001}
M.~Kearns, M.~L. Littman, and S.~Singh.
\newblock Graphical models for game theory.
\newblock In \emph{Proceedings of the Seventeenth Conference on Uncertainty in
  Artificial Intelligence (UAI2001)}, pages 253--260, 2001.

\bibitem[Nisan and Roughgarden(2007)]{Nisan.ea:2007}
N.~Nisan and T.~Roughgarden.
\newblock \emph{Algorithmic Game Theory}.
\newblock Cambridge University Press, 2007.

\bibitem[Daskalakis and Papadimitriou(2006)]{Daskalakis.Papadimitriou:2006}
C.~Daskalakis and C.H. Papadimitriou.
\newblock Computing pure {N}ash equilibria in graphical games via {M}arkov
  random fields.
\newblock In \emph{Proceedings of the 7th ACM Conference on Electronic
  Commerce}, pages 91--99, 2006.

\bibitem[Papadimitriou and Roughgarden(2008)]{Papadimitriou.Roughgarden:2008}
C.~H. Papadimitriou and T.~Roughgarden.
\newblock Computing correlated equilibria in multi-player games.
\newblock \emph{Journal of the ACM}, 55\penalty0 (3):\penalty0 1--29, 2008.

\bibitem[Babichenko and Tamuz(2016)]{Babichenko.Tamuz:2016}
Y.~Babichenko and O.~Tamuz.
\newblock Graphical potential games.
\newblock \emph{Journal of Economic Theory}, 163:\penalty0 889--899, 2016.

\bibitem[Lauritzen(1996)]{Lauritzen:1996}
S.~L. Lauritzen.
\newblock \emph{Graphical models}.
\newblock Oxford University Press, 1996.

\bibitem[Candogan et~al.(2011)Candogan, Menache, Ozdaglar, and
  Parrilo]{Candogan.ea:2011}
O.~Candogan, I.~Menache, A.~Ozdaglar, and P.A. Parrilo.
\newblock Flows and decompositions of games: Harmonic and potential games.
\newblock \emph{Mathematics of Operations Research}, 36\penalty0 (3):\penalty0
  474--503, 2011.

\bibitem[Bretto(2013)]{Bretto:2013}
A.~Bretto.
\newblock \emph{Hypergraph Theory: An Introduction}.
\newblock Springer, 2013.

\bibitem[Gallo et~al.(1993)Gallo, Longo, Pallottino, and Nguyen]{Gallo.ea:1993}
G.~Gallo, G.~Longo, S.~Pallottino, and S.~Nguyen.
\newblock Directed hypergraphs and applications.
\newblock \emph{Discrete Applied Mathematics}, 42\penalty0 (2):\penalty0
  177--201, 1993.

\bibitem[Monderer and Shapley(1996)]{Monderer.Shapley:1996}
D.~Monderer and L.~Shapley.
\newblock Potential games.
\newblock \emph{Games and Economic Behavior}, 14:\penalty0 124--143, 1996.

\bibitem[Daskalakis and Papadimitriou(2009)]{Daskalakis.Papadimitriou:2009}
C.~Daskalakis and C.H. Papadimitriou.
\newblock On a network generalization of the minmax theorem.
\newblock In \emph{Proceedings of the 36th Internatilonal Colloquium on
  Automata, Languages and Programming (ICALP '09): Part II}, pages 423--434,
  2009.

\bibitem[Cai and Daskalakis(2011)]{Cai.Daskalakis:2011}
Y.~Cai and C.~Daskalakis.
\newblock On minmax theorems for multiplayer games.
\newblock In \emph{Proceeding Proceedings of the twenty-second annual ACM-SIAM
  symposium on Discrete algorithms (SODA '11)}, pages 217--234, 2011.

\bibitem[Yanovskaya(1968)]{Yanovskaya:1968}
E.~B. Yanovskaya.
\newblock Equilibrium points in polymatrix games.
\newblock \emph{Litovskii Matematicheskii Sbornik}, 8:\penalty0 381--384, 1968.

\bibitem[Ortiz and Irfan(2017)]{Ortiz.Irfan:2017}
Luis~E. Ortiz and M.~T. Irfan.
\newblock Fptas for mixed-strategy nash equilibria in tree graphical games and
  their generalizations.
\newblock In \emph{Proceedings of the Thirty-First AAAI Conference on
  Artificial Intelligence (AAAI-17)}, 2017.

\end{thebibliography}


\end{document}